\documentclass[oneside]{amsart}
\usepackage[T1]{fontenc}
\usepackage[latin9]{inputenc}
\usepackage{geometry}
\usepackage{color}
\usepackage{amsthm}
\usepackage{amstext}
\usepackage{amssymb}
\usepackage{graphicx}

\makeatletter
\numberwithin{equation}{section}
\numberwithin{figure}{section}
  \theoremstyle{definition}
  \newtheorem{problem}{\protect\problemname}
  \theoremstyle{plain}
  \newtheorem{fact}{\protect\factname}
  \theoremstyle{plain}
  \newtheorem*{thm*}{\protect\theoremname}
  \theoremstyle{plain}
  \newtheorem{prop}{\protect\propositionname}
  \theoremstyle{remark}
  \newtheorem*{rem*}{\protect\remarkname}
  \theoremstyle{definition}
  \newtheorem{defn}{\protect\definitionname}
\theoremstyle{plain}
\newtheorem{thm}{\protect\theoremname}

\usepackage[foot]{amsaddr}
\makeatletter
   \renewcommand{\subsection}{\@startsection{subsection}{1}{0mm}
   {\baselineskip}%
   {0.5\baselineskip}{\normalfont\normalsize\bfseries}}%
   \makeatother
\makeatletter
   \renewcommand{\subsubsection}{\@startsection{subsubsection}{1}{0mm}
   {\baselineskip}%
   {0.5\baselineskip}{\normalfont\normalsize\itshape}}%
   \makeatother
\usepackage{enumerate}
\makeatother

  \providecommand{\definitionname}{Definition}
  \providecommand{\factname}{Fact}
  \providecommand{\problemname}{Problem}
  \providecommand{\propositionname}{Proposition}
  \providecommand{\remarkname}{Remark}
  \providecommand{\theoremname}{Theorem}
\providecommand{\theoremname}{Theorem}

\begin{document}
\title{An Algorithm for the T-count}
\author{David Gosset$^{\dagger}$}
\author{Vadym Kliuchnikov$^{\star}$}
\author{Michele Mosca$^{\dagger,\triangle}$}
\author{Vincent Russo$^{\star}$}
\address{$^\dagger$ Institute for Quantum Computing and Department of Combinatorics \& Optimization, University of Waterloo} \address{$^\star$ Institute for Quantum Computing and David R. Cheriton School of Computer Science, University of Waterloo}
\address{$^\triangle$ Perimeter Institute for Theoretical Physics, Waterloo, ON}
\begin{abstract}
We consider quantum circuits composed of Clifford and $T$ gates. In this context the $T$ gate has a special status since it confers universal computation when added to the (classically simulable) Clifford gates. However it can be very expensive to implement fault-tolerantly. We therefore view this gate as a resource which should be used only when necessary. Given an $n$-qubit unitary $U$ we are interested in computing a circuit that implements it using the minimum possible number of $T$ gates (called the $T$-count of $U$). A related task is to decide if the $T$-count of $U$ is less than or equal to $m$; we consider this problem as a function of $N=2^n$ and $m$. We provide a classical algorithm which solves it using time and space both upper bounded as $\mathcal{O}(N^m \text{poly}(m,N))$. We implemented our algorithm and used it to show that any Clifford+T circuit for the Toffoli or the Fredkin gate requires at least 7 $T$ gates. This implies that the known 7 $T$ gate circuits for these gates are $T$-optimal. We also provide a simple expression for the $T$-count of single-qubit unitaries.
\end{abstract}
\maketitle

\section{Introduction}

The single-qubit $T$ gate 
\[
T=\left(\begin{array}{cc}
1 & 0\\
0 & e^{i\frac{\pi}{4}}
\end{array}\right),
\]
along with all gates from the Clifford group, is a universal gate
set for quantum computation. The $T$ gate is essential because circuits
composed of only Clifford gates are classically simulable \cite{Gottesman_Heisenberg}.
The $T$ gate also plays a special role in fault-tolerant quantum
computation. In contrast with Clifford gates, the $T$ gate is not
transversal for many quantum error-correcting codes, which means that
in practice it is very costly to implement fault-tolerantly. For this
reason we are interested in circuits which use as few $T$ gates as
possible to implement a given unitary. 

We consider unitaries which can be implemented exactly using Clifford
and T gates. The set of all such unitaries is known: it was conjectured
in \cite{kliuchnikov2012fast} and later proven in \cite{giles2013exact}
that an $n$-qubit unitary can be implemented using this gate set
if and only if its matrix elements are in the ring $\mathbb{Z}\left[i,\frac{1}{\sqrt{2}}\right]$.
In general this implementation requires one ancilla qubit prepared
in the state $|0\rangle$ in addition to the $n$ qubits on which
the computation is to be performed. An algorithm for exactly synthesizing
unitaries over the gate set $\{H,T,\text{CNOT}\}$ was given in reference
\cite{giles2013exact}, and a superexponentially faster version of
this algorithm was presented in reference \cite{kliuchnikov2013synthesis}. 

In this paper we focus on the set of unitaries implementable without
ancillas, that is to say, the group $\mathcal{J}_{n}$ generated by
the $n$-qubit Clifford group and the T gate. An $n$-qubit unitary
is an element of this group if and only if its matrix elements are
in the ring $\mathbb{Z}\left[i,\frac{1}{\sqrt{2}}\right]$ and its
determinant satisfies a simple condition \cite{giles2013exact} (when
$n\geq4$ the condition is that the determinant is equal to $1$).
Notable examples include the Toffoli and Fredkin gates which are in
the group $\mathcal{J}_{3}$. 

For $U\in\mathcal{J}_{n}$, the $T$-count of $U$ is defined to be
the minimum number of $T$ gates in a Clifford+T circuit that implements
it (up to a possible global phase), and is denoted by $\mathcal{T}(U)$.
In other words $\mathcal{T}(U)$ is the minimum $m$ for which 
\begin{equation}
e^{i\phi}U=C_{m}T_{(q_{m})}C_{m-1}T_{(q_{m-1})}\ldots T_{(q_{1})}C_{0}\label{eq:clifford_T_expansion}
\end{equation}
where $\phi\in[0,2\pi)$, $C_{i}$ are in the $n$-qubit Clifford
group, $q_{j}\in\{1,\ldots,n\}$, and $T_{(r)}$ indicates the $T$
gate acting on the $r$th qubit. 

Note that it may be possible to implement some unitary using less
than $\mathcal{T}(U)$ $T$ gates if one uses ancilla qubits and/or
measurements with classically controlled operations. For example,
Jones has shown how to perform a Toffoli gate using these additional
ingredients and only four $T$ gates \cite{Jones2013}. This does
not contradict our result that $\mathcal{T}(\text{Toffoli)=7}$.

We are interested in the following problem. Given $U\in\mathcal{J}_{n}$,
compute a $T$-optimal $n$-qubit quantum circuit for it, that is
to say, a circuit which implements it (up to a global phase) using
$\mathcal{T}(U)$ $T$ gates. A related, potentially easier problem,
is to compute $\mathcal{T}(U)$ given $U$. It turns out that this
latter problem is not much easier: we show in Section \ref{sec:Unitary-Decomposition}
that an algorithm which computes $\mathcal{T}(U)$ can be converted
into an algorithm which outputs a T-optimal circuit for $U$, with
overhead polynomial in $\mathcal{T}(U)$ and the dimension of $U$.
For this reason, and for the sake of simplicity, we focus on the task
of computing $\mathcal{T}(U)$.
\begin{problem}
[\textbf{COUNT-T}] \label{COUNT-T}Given $U\in\mathcal{J}_{n}$ and
$m\in\mathbb{N}$, decide if $\mathcal{T}(U)\leq m$.
\end{problem}
We consider the complexity of this problem as a function of $m$ and
$N=2^{n}$. We treat arithmetic operations on the entries of $U$
at unit cost, and we do not account for the bit-complexity associated
with specifying or manipulating them. We present an algorithm which
solves COUNT-T using time and space both upper bounded as $\mathcal{O}\left(N^{m}\text{poly}(m,N)\right)$.
Our algorithm uses the meet-in-the-middle idea from \cite{amy2013meet}
along with a representation for Clifford+T unitaries where the $T$
gates have a special role. We implemented our algorithm in C++ and
used it to prove that the $T$-count of Toffoli and Fredkin gates
is $7$. We also provide a simple expression for the $T$-count of
single-qubit unitaries. 

The problem of computing $T$-optimal circuits was studied in references
\cite{kliuchnikov2012fast,amy2013meet,matroid}. In reference \cite{kliuchnikov2012fast}
an algorithm was given for synthesizing single-qubit circuits over
the gate set $\{H,T\}$, and it was shown that the resulting circuits
are $T$-optimal. The problem of reducing the number of $T$-gates
in circuits with $n>1$ qubits was considered as an application of
the meet-in-the-middle algorithm from reference \cite{amy2013meet},
where some small examples of $T$-depth optimal circuits were found.
The algorithm for optimizing $T$-depth presented in reference \cite{amy2013meet}
can be used (with a small modification) to solve COUNT-T but its time
complexity is $\Omega\left(2^{n^{2}\left\lfloor m/2\right\rfloor }\right)$
since the size of the $n$-qubit Clifford group is $\Omega\left(2^{n^{2}}\right)$\cite{Calderbank}.
It was conjectured in reference \cite{amy2013meet} that Toffoli requires
7 $T$ gates; we prove this conjecture in this paper. In reference
\cite{matroid} an algorithm based on matroid partitioning is given
which can be used as a heuristic for minimizing the $T$-count and
$T$-depth of quantum circuits. The algorithm we present here solves
COUNT-T using time $\mathcal{O}\left(2^{nm}\text{poly}\left(m,2^{n}\right)\right)$,
which is a superexponential speedup (as a function of $n$) over reference
\cite{amy2013meet}, and in contrast with the heuristic from \cite{matroid}
it computes the $T$-count exactly. 

Our work, and previous work on the $T$-count, takes a different but
complementary perspective from that of the recent paper \cite{Veitch2013}.
That paper considers quantum states (e.g., magic states) as a resource
for computation using Clifford circuits, and attempts to quantify
the amount of resource in a given state. Here we view the $T$ gate
as a resource and ask how much is necessary to implement a given unitary.

The rest of the paper is structured as follows. We begin in Section
\ref{sec:Preliminaries} by discussing the central objects that we
study: Pauli and Clifford operators, and the group $\mathcal{J}_{n}$
generated by Clifford and $T$ gates. In Section \ref{sec:Unitary-Decomposition}
we give a decomposition for unitaries from the group $\mathcal{J}_{n}$.
Using this decomposition we show how an algorithm for the $T$-count
can be used to generate $T$-optimal circuits. In Section \ref{sec:-count-for-single-qubit}
we give a simple characterization of the $T$-count for single-qubit
unitaries, and in Section \ref{sec:Algorithm-for-the} we present
our algorithm for COUNT-T. We conclude in Section \ref{sec:Conclusions-and-open}
with a discussion of open problems.

\section{Preliminaries\label{sec:Preliminaries}}

In this Section we establish notation and we review facts about the
Clifford+T gate set. Throughout this paper we write $N=2^{n}$ and
$[K]=\{1,\ldots,K\}$. We write 
\[
X=\left(\begin{array}{cc}
0 & 1\\
1 & 0
\end{array}\right)\qquad Y=\left(\begin{array}{cc}
0 & -i\\
i & 0
\end{array}\right)\qquad Z=\left(\begin{array}{cc}
1 & 0\\
0 & -1
\end{array}\right)
\]
 for the single-qubit Pauli matrices. We use a parenthesized subscript
to indicate qubits on which an operator acts, e.g., $X_{(1)}=X\otimes\mathbb{I}^{\otimes\left(n-1\right)}$
indicates the Pauli $X$ matrix acting on the first qubit.

\subsection{Cliffords and Paulis}

The single-qubit Clifford group $\mathcal{C}_{1}$ is generated by
the Hadamard and phase gates
\[
\mathcal{C}_{1}=\left\langle H,T^{2}\right\rangle 
\]
 where
\[
H=\frac{1}{\sqrt{2}}\left(\begin{array}{cc}
1 & 1\\
1 & -1
\end{array}\right)\quad T^{2}=\left(\begin{array}{cc}
1 & 0\\
0 & i
\end{array}\right).
\]
When $n>1$, the $n$-qubit Clifford group $\mathcal{C}_{n}$ is generated
by these two gates (acting on any of the $n$ qubits) along with the
two-qubit $\text{CNOT}=|0\rangle\langle0|\otimes\mathbb{I}+|1\rangle\langle1|\otimes X$
gate (acting on any pair of qubits). The Clifford group is special
because of its relationship to the set of $n$-qubit Pauli operators
\[
\mathcal{P}_{n}=\left\{ Q_{1}\otimes Q_{2}\otimes\ldots\otimes Q_{n}:\; Q_{i}\in\{\mathbb{I},X,Y,Z\}\right\} .
\]

Cliffords map Paulis to Paulis, up to a possible phase of $-1$, i.e.,
for any $P\in\mathcal{P}_{n}$ and any $C\in\mathcal{C}_{n}$ we have
\[
CPC^{\dagger}=(-1)^{b}P^{\prime}
\]
 for some $b\in\{0,1\}$ and $P^{\prime}\in\mathcal{P}_{n}$. It is
easy to prove that, given two Paulis, neither equal to the identity,
it is always possible to (efficiently) find a Clifford which maps
one to the other. For completeness we include a proof in the Appendix.
\begin{fact}
\label{Clif_Fact}For any $P,P^{\prime}\in\mathcal{P}_{n}\setminus\{\mathbb{I}\}$
there exists a Clifford $C\in\mathcal{C}_{n}$ such that $CPC^{\dagger}=P^{\prime}$.
A circuit for $C$ over the gate set $\{H,T^{2},\text{CNOT}\}$ can
be computed efficiently (as a function of $n$).
\end{fact}

\subsection{The group $\mathcal{J}_{n}$ generated by Clifford and T gates}

We consider the group $\mathcal{J}_{n}$ generated by the $n$-qubit
Clifford group along with the $T$ gate. For a single qubit 
\[
\mathcal{J}_{1}=\left\langle H,T\right\rangle ,
\]
 and for $n>1$ qubits 
\[
\mathcal{J}_{n}=\left\langle H_{(i)},T_{(i)},\text{CNOT}_{(i,j)}:\; i,j\in[n]\right\rangle 
\]
 (recall the subscript indicates qubits on which the gate acts). It
is not hard to see that $\mathcal{J}_{n}$ is a group, since the Hadamard
and $\text{CNOT}$ gates are their own inverses and $T^{-1}=T^{7}$. 

Noting that $H,T,$ and $\text{CNOT}$ all have matrix elements over
the ring 
\[
\mathbb{Z}\left[i,\frac{1}{\sqrt{2}}\right]=\left\{ \frac{a+bi+c\sqrt{2}+di\sqrt{2}}{\sqrt{2}^{k}}:\; a,b,c,d\in\mathbb{Z},\; k\in\mathbb{N}\right\} ,
\]
we see that any $U\in\mathcal{J}_{n}$ also has matrix elements from
this ring. Giles and Selinger \cite{giles2013exact} proved that this
condition, along with a simple condition on the determinant, exactly
characterizes $\mathcal{J}_{n}$.
\begin{thm*}
[Corollary 2 from \cite{giles2013exact}] 
\[
\mathcal{J}_{n}=\left\{ U\in U(N):\:\text{Each entry of \ensuremath{U}is an element of }\mathbb{Z}\left[i,\frac{1}{\sqrt{2}}\right],\text{ and }\det U=e^{i\frac{\pi}{8}Nr}\text{ for some }r\in[8]\right\} 
\]

where $N=2^{n}.$
\end{thm*}
Note that for $n\geq4$ the condition on the determinant is simply
$\det U=1$.

\subsection{Channel representation, and modding out global phases}

Consider the action of an $n$-qubit unitary $U$ on a Pauli $P_{s}\in\mathcal{P}_{n}$
\begin{equation}
UP_{s}U^{\dagger}.\label{eq:U_conjugation}
\end{equation}
The set of all such operators (with $P_{s}\in\mathcal{P}_{n}$) completely
determines $U$ up to a global phase. Since $\mathcal{P}_{n}$ is
a basis for the space of all Hermitian $N\times N$ matrices, we can
expand (\ref{eq:U_conjugation}) as
\begin{equation}
UP_{s}U^{\dagger}=\sum_{P_{r}\in\mathcal{P}_{n}}\widehat{U}_{rs}P_{r},\label{eq:UPU}
\end{equation}
 where 
\begin{equation}
\widehat{U}_{rs}=\frac{1}{2^{n}}\text{Tr}\left(P_{r}UP_{s}U^{\dagger}\right).\label{eq:U_hat_mtx_el}
\end{equation}
 This defines a $N^{2}\times N^{2}$ matrix $\widehat{U}$ with rows
and columns indexed by Paulis $P_{r},P_{s}\in\mathcal{P}_{n}$. We
refer to $\widehat{U}$ as the channel representation of $U$. 

By Hermitian-conjugating (\ref{eq:UPU}) we see that each matrix element
of $\widehat{U}$ is real. It is also straightforward to check, using
(\ref{eq:UPU}), that the channel representation respects matrix multiplication:
\[
\widehat{UV}=\widehat{U}\widehat{V}.
\]
Setting $V=U^{\dagger}$ and using the fact that $\widehat{U^{\dagger}}=\left(\widehat{U}\right)^{\dagger}$,
we see that the channel representation $\widehat{U}$ is unitary.

In this paper we use the channel representation for $U\in\mathcal{J}_{n}$.
In this case, the entries of $U$ are in the ring $\mathbb{Z}\left[i,\frac{1}{\sqrt{2}}\right]$,
and from (\ref{eq:U_hat_mtx_el}) we see that so are the entries of
$\widehat{U}$. Since $\widehat{U}$ is also real, its entries are
from the subring 
\[
\mathbb{Z}\left[\frac{1}{\sqrt{2}}\right]=\left\{ \frac{a+b\sqrt{2}}{\sqrt{2}^{k}}:\; a,b\in\mathbb{Z},\; k\in\mathbb{N}\right\} .
\]

The channel representation identifies unitaries which differ by a
global phase. We write 
\begin{align*}
\widehat{\mathcal{J}_{n}} & =\left\{ \widehat{U}:\; U\in\mathcal{J}_{n}\right\} ,\qquad\widehat{\mathcal{C}_{n}}=\left\{ \widehat{C}:\; C\in\mathcal{C}_{n}\right\} 
\end{align*}
for the groups in which global phases are modded out. Note that each
$Q\in\widehat{\mathcal{C}_{n}}$ is a unitary matrix with one nonzero
entry in each row and each column, equal to $\pm1$ (since Cliffords
map Paulis to Paulis up to a possible phase of $-1$). The converse
also holds: if $W\in\widehat{\mathcal{J}_{n}}$ has this property
then $W\in\widehat{\mathcal{C}_{n}}$. 

Finally, note that since the definition of $T$-count is insensitive
to global phases, it is well-defined in the channel representation;
for $U\in\mathcal{J}_{n}$ we define $\mathcal{T}(\widehat{U})=\mathcal{T}(U)$.

\section{decomposition for unitaries in $\mathcal{J}_{n}$\label{sec:Unitary-Decomposition}}

In this Section we present a decomposition for unitaries over the
Clifford+$T$ gate set and we use it to show that an algorithm for
COUNT-T can be used to generate a $T$-optimal circuit for $U\in\mathcal{J}_{n}$
with overhead polynomial in $\mathcal{T}(U)$ and $N$.

By definition, any $U\in\mathcal{J}_{n}$ can be written as an alternating
product of Cliffords and T gates as in equation (\ref{eq:clifford_T_expansion}).
The number $m$ of $T$ gates in (\ref{eq:clifford_T_expansion})
can be taken to be $\mathcal{T}(U)$, and we can reorganize it so
that each $T$ gate is conjugated by a Clifford: 
\begin{equation}
U=e^{i\phi}\left(\prod_{i=\mathcal{T}(U)}^{1}D_{i}T_{(q_{i})}D_{i}^{\dagger}\right)D_{0}\label{eq:T1_eqn}
\end{equation}
where $D_{i}=\prod_{j=\mathcal{T}(U)}^{i}C_{j}$ for $i\in\{0,1,\ldots,\mathcal{T}(U)\}$.
Note that 
\[
D_{i}T_{(q_{i})}D_{i}^{\dagger}=\frac{1}{2}\left(1+e^{i\frac{\pi}{4}}\right)\mathbb{I}+\frac{1}{2}\left(1-e^{i\frac{\pi}{4}}\right)D_{i}Z_{(q_{i})}D_{i}^{\dagger}.
\]
 In the second term we have a Pauli operator conjugated by a Clifford,
which is equal to another Pauli (up to a possible phase of $-1$).
In other words (\ref{eq:T1_eqn}) can be written
\begin{equation}
U=e^{i\phi}\left(\prod_{i=\mathcal{T}(U)}^{1}R((-1)^{b_{i}}P_{i})\right)D_{0}\label{eq:with_generalbi}
\end{equation}
 where $P_{i}\in\mathcal{P}_{n}\setminus\{\mathbb{I}\}$ , $b_{i}\in\{0,1\}$,
and 
\[
R(\pm P)=e^{i\frac{\pi}{8}\left(\mathbb{I}\mp P\right)}=\frac{1}{2}\left(1+e^{i\frac{\pi}{4}}\right)\mathbb{I}\pm\frac{1}{2}\left(1-e^{i\frac{\pi}{4}}\right)P\qquad P\in\mathcal{P}_{n}\setminus\{\mathbb{I}\}.
\]
The following Proposition shows that such a decomposition exists with
$b_{i}=0$ for all $i\in[\mathcal{T}(U)]$.
\begin{prop}
\label{decomposition}For any $U\in\mathcal{J}_{n}$ there exists
a phase $\phi\in[0,2\pi),$ a Clifford $C_{0}\in\mathcal{C}_{n}$
and Paulis $P_{i}\in\mathcal{P}_{n}\setminus\{\mathbb{I}\}$ for $i\in[\mathcal{T}(U)]$
such that
\begin{equation}
U=e^{i\phi}\left(\prod_{i=\mathcal{T}(U)}^{1}R(P_{i})\right)C_{0}.\label{eq:all_bis_zero}
\end{equation}
\end{prop}
\begin{proof}
For any $Q\in\mathcal{P}_{n}\setminus\{\mathbb{I}\}$ we have 
\begin{equation}
R(-Q)=e^{i\frac{\pi}{8}\left(1+Q\right)}=R(Q)e^{i\frac{\pi}{4}Q}.\label{eq:R_Q_minusQ}
\end{equation}
Note, using Fact \ref{Clif_Fact}, that $e^{i\frac{\pi}{4}Q}=Ce^{i\frac{\pi}{4}Z_{(1)}}C^{\dagger}$
for some $C\in\mathcal{C}_{n}$, and 
\[
e^{i\frac{\pi}{4}Z_{(1)}}=e^{i\frac{\pi}{4}}\left(T_{(1)}\right)^{6}.
\]
Now using the fact that $T^{6}$ is Clifford we get $e^{-i\frac{\pi}{4}}e^{i\frac{\pi}{4}Q}\in\mathcal{C}_{n}$.

By repeatedly using (\ref{eq:R_Q_minusQ}) we transform (\ref{eq:with_generalbi})
into the form (\ref{eq:all_bis_zero}) (with a potentially different
set of Paulis $P_{i}$). Start with the decomposition (\ref{eq:with_generalbi})
and let $j\in[\mathcal{T}(U)]$ be the largest index such that $b_{j}=1$.
Then 
\begin{align}
U & =e^{i\phi}\left(\prod_{i=\mathcal{T}(U)}^{j+1}R\left(P_{i}\right)\right)R(-P_{j})\left(\prod_{i=j-1}^{1}R\left((-1)^{b_{i}}P_{i}\right)\right)D_{0}\nonumber \\
 & =e^{i\phi}\left(\prod_{i=\mathcal{T}(U)}^{j+1}R\left(P_{i}\right)\right)R(P_{j})e^{i\frac{\pi}{4}P_{j}}\left(\prod_{i=j-1}^{1}R\left((-1)^{b_{i}}P_{i}\right)\right)D_{0}\nonumber \\
 & =e^{i\left(\phi+\frac{\pi}{4}\right)}\left(\prod_{i=\mathcal{T}(U)}^{j}R\left(P_{i}\right)\right)\left[e^{-i\frac{\pi}{4}}e^{i\frac{\pi}{4}P_{j}}\left(\prod_{i=j-1}^{1}R\left((-1)^{b_{i}}P_{i}\right)\right)D_{0}\right].\label{eq:P_j_bi_zero}
\end{align}
Note that all the Paulis outside of the square brackets have $+$
signs in front of them. The unitary in square brackets has $T$-count
$j-1$ (since $e^{-i\frac{\pi}{4}}e^{i\frac{\pi}{4}P_{j}}\in\mathcal{C}_{n}$)
and can therefore be written as in equation (\ref{eq:with_generalbi})
with $j-1$ terms in the product. We now recurse, applying the above
steps to this unitary in square brackets. We repeat this procedure
(at most $j-1$ more times) until all the minus signs are replaced
with plus signs. \end{proof}
\begin{rem*}
The channel representation $\widehat{U}$ inherits the decomposition
from Proposition \ref{decomposition}, and in this representation
the global phase goes away, i.e., 
\begin{equation}
\widehat{U}=\left(\prod_{i=\mathcal{T}(U)}^{1}\widehat{R(P_{i})}\right)\widehat{C_{0}}.\label{eq:canon_channel_rep}
\end{equation}

\end{rem*}

\subsection*{Computing $T$-optimal circuits using an algorithm for COUNT-T}

Suppose that we have an algorithm $\mathcal{A}$ which solves the
decision problem COUNT-T. For any $U\in\mathcal{J}_{n}$, we show
that, with overhead polynomial in $N$ and $\mathcal{T}(U)$, such
an algorithm can also be used to generate a $T$-optimal circuit for
$U$ over the gate set $\{H,T,CNOT\}$. This is a simple application
of the decomposition (\ref{eq:all_bis_zero}). 

First note that $\mathcal{A}$ can be used to compute $\mathcal{T}(U)$,
by starting with $m=0$ and running $\mathcal{A}$ on each nonnegative
integer in sequence until we find $m-1$ and $m$ with $\mathcal{T}(U)\leq m$
and $\mathcal{T}(U)>m-1$. In this way we first compute $\mathcal{T}(U)$
and then we try each Pauli $Q\in\mathcal{P}_{n}\setminus\{\mathbb{I}\}$
until we find one which satisfies 
\[
\mathcal{T}\left(R(Q)^{\dagger}U\right)=\mathcal{T}(U)-1.
\]
 Note that (\ref{eq:all_bis_zero}) implies that at least one such
Pauli $Q$ exists. Repeating this procedure $\mathcal{T}(U)-1$ more
times we get Paulis $Q_{1},\ldots,Q_{\mathcal{T}(U)}$ such that 
\[
\mathcal{T}\left(R(Q_{1})^{\dagger}\ldots R(Q_{\mathcal{T}(U)-1})^{\dagger}R(Q_{\mathcal{T}(U)})^{\dagger}U\right)=0.
\]
In other words 
\begin{equation}
e^{i\phi}U=R(Q_{\mathcal{T}(U)})R(Q_{\mathcal{T}(U)-1})\ldots R(Q_{1})C_{0}\label{eq:U_R_Q}
\end{equation}
where $\phi\in[0,2\pi)$ and $C_{0}\in\mathcal{C}_{n}$. 

The next step is to compute a circuit which implements $C_{0}$ up
to a global phase. To this end, we compute the channel representation
$\widehat{C_{0}}$ explicitly as an $N^{2}\times N^{2}$ matrix using
the fact that 
\[
\widehat{C_{0}}=\widehat{R(Q_{1})^{\dagger}}\ldots\widehat{R(Q_{\mathcal{T}(U)-1})^{\dagger}}\widehat{R(Q_{\mathcal{T}(U)})^{\dagger}}\widehat{U}.
\]
We then use standard techniques to obtain a circuit which implements
$C_{0}$ (up to a global phase) over the gate set $\{H,T^{2},\text{CNOT}\}$.
For example, this can be done efficiently using the procedure outlined
in the proof of Theorem 8 from reference \cite{aaronson2004improved}.

Now using Fact \ref{Clif_Fact} there are Cliffords $C_{1},\ldots,C_{\mathcal{T}(U)}\in\mathcal{C}_{n}$
which map each of the Paulis $Q_{1},\ldots,Q_{\mathcal{T}(U)}$ to
$Z$ acting on the first qubit, i.e., 
\[
C_{i}Q_{i}C_{i}^{\dagger}=Z_{(1)}
\]
 for each $i\in[\mathcal{T}(U)]$, and we can efficiently compute
circuits for each of these Cliffords. Using the fact that $R(Z_{(1)})=T_{(1)}$
and plugging into equation (\ref{eq:U_R_Q}) gives 
\[
e^{i\phi}U=\left(\prod_{j=\mathcal{T}(U)}^{1}C_{j}^{\dagger}T_{(1)}C_{j}\right)C_{0}.
\]
The RHS, along with the circuits for $C_{0},\ldots,C_{\mathcal{T}(U)}$
discussed above, gives a $T$-optimal implementation of $U$ (up to
a global phase) over the gate set $\{H,T,\text{CNOT\}}$.

\section{$T$-count for single-qubit unitaries\label{sec:-count-for-single-qubit}}

In this Section we show how the $T$-count of a single-qubit unitary
$U\in\mathcal{J}_{1}$ can be directly computed from its channel representation
$\widehat{U}$.

Recall that $\widehat{U}$ has entries over the ring $\mathbb{Z}\left[\frac{1}{\sqrt{2}}\right]$.
For any nonzero element $\frac{a+b\sqrt{2}}{\sqrt{2}^{k}}$ of this
ring, the integer $k$ can be chosen to be minimal in the following
sense. If $a$ is even and $k>0$ then we can divide the top and bottom
by $\sqrt{2}$ and reduce $k$ by one. When $a$ is odd or $k=0$,
$k$ cannot be reduced any further, and we call it the \textit{smallest
denominator exponent. }Similar quantities were defined in references
\cite{kliuchnikov2012fast} and \cite{giles2013exact}.
\begin{defn}
\label{sde_defn}For any nonzero $v\in\mathbb{\mathbb{Z}}\left[\frac{1}{\sqrt{2}}\right]$
the smallest denominator exponent, denoted by $\text{sde}(v)$, is
the smallest $k\in\mathbb{N}$ for which 
\[
v=\frac{a+b\sqrt{2}}{\sqrt{2}^{k}}
\]
 with $a,b\in\mathbb{Z}$. We define $\mathrm{sde}(0)=0$. For a $d_{1}\times d_{2}$
matrix $M$ with entries over this ring we define
\[
\mathrm{sde}\left(M\right)=\max_{a\in[d_{1}],b\in[d_{2}]}\mathrm{sde}(M_{ab}).
\]

\end{defn}
We use the following fact which is straightforward to prove.
\begin{fact}
\label{Fact}Let $q,r\in\mathbb{\mathbb{Z}}\left[\frac{1}{\sqrt{2}}\right]$
with $\mathrm{sde}(q)>\mathrm{sde}(r)$. Then 
\[
\mathrm{sde}\left(\frac{1}{\sqrt{2}}\left(q\pm r\right)\right)=\mathrm{sde}(q)+1.
\]
 
\end{fact}
The $T$-count of a single qubit unitary is simply equal to the $\mathrm{sde}$
of its channel representation.
\begin{thm}
The $T$-count of a single-qubit unitary $U\in\mathcal{J}_{1}$ is
\[
\mathcal{T}(U)=\mathrm{sde}(\widehat{U}).
\]
\end{thm}
\begin{proof}
By Proposition \ref{decomposition}, $U$ can be decomposed as (\ref{eq:all_bis_zero}),
as a global phase times a product of $R(P_{i})$ operators times $C_{0}$,
with $P_{i}\in\{X,Y,Z\}$ for all $i$ and $C_{0}\in\mathcal{C}_{1}$.
Note $\mathcal{T}\left(U\right)=\mathcal{T}(UC_{0}^{\dagger})$ and
\begin{equation}
\widehat{UC_{0}^{\dagger}}=\prod_{i=\mathcal{T}(U)}^{1}\widehat{R(P_{i})}.\label{eq:UC_channel}
\end{equation}
The operators which appear on the RHS are 
\begin{equation}
\widehat{R(X)}=\left(\begin{array}{cccc}
1 & 0 & 0 & 0\\
0 & 1 & 0 & 0\\
0 & 0 & \frac{1}{\sqrt{2}} & -\frac{1}{\sqrt{2}}\\
0 & 0 & \frac{1}{\sqrt{2}} & \frac{1}{\sqrt{2}}
\end{array}\right)\quad\widehat{R(Y)}=\left(\begin{array}{cccc}
1 & 0 & 0 & 0\\
0 & \frac{1}{\sqrt{2}} & 0 & \frac{1}{\sqrt{2}}\\
0 & 0 & 1 & 0\\
0 & -\frac{1}{\sqrt{2}} & 0 & \frac{1}{\sqrt{2}}
\end{array}\right)\quad\widehat{R(Z)}=\left(\begin{array}{cccc}
1 & 0 & 0 & 0\\
0 & \frac{1}{\sqrt{2}} & -\frac{1}{\sqrt{2}} & 0\\
0 & \frac{1}{\sqrt{2}} & \frac{1}{\sqrt{2}} & 0\\
0 & 0 & 0 & 1
\end{array}\right)\label{eq:three_matrices}
\end{equation}
 where the rows and columns are labeled by Paulis $\mathrm{\mathbb{I},X,Y,Z}$
from top to bottom and left to right. Note that 
\[
\widehat{R(X)}^{2}=\widehat{R(Y)}^{2}=\widehat{R(Z)}^{2}=\mathbb{I}.
\]
Since (\ref{eq:UC_channel}) uses the minimal number of $T$ gates,
we see that no two consecutive Paulis in this decomposition are equal,
i.e., $P_{i}\neq P_{i+1}$ for all $i\in[\mathcal{T}(U)-1]$. 

Consider the entries of (\ref{eq:UC_channel}). Since it is a product
of the matrices (\ref{eq:three_matrices}), we see that the top left
entry is always $1$ and all other entries in the first row and column
are $0$. We therefore focus on the lower-right $3\times3$ submatrix.
Looking at this submatrix for $\widehat{R(X)}$, $\widehat{R(Y)}$
and $\widehat{R(Z)}$ we see that there are two rows with $\mathrm{sde}$
equal to $1$ and one row where it is $0$ (recall from Definition
\ref{sde_defn} that the $\mathrm{sde}$ of a row vector is the maximum
$\mathrm{sde}$ of one of its entries). The rows with $\mathrm{sde}$
$0$ in the lower-right $3\times3$ submatrix of $\widehat{R(X)}$,
$\widehat{R(Y)}$ and $\widehat{R(Z)}$ are labeled by $X$, $Y$
and $Z$ respectively. Taking this as the base case, we prove by induction
that the lower right $3\times3$ submatrix of (\ref{eq:UC_channel})
contains two rows with $\mathrm{sde}$ $\mathcal{T}(U)$ and one row
with $\mathrm{sde}$ $\mathcal{T}(U)-1$, and this latter row is the
one labeled by the Pauli $P_{\mathcal{T}(U)}$. We suppose this holds
for $\mathcal{T}(U)=k$ and we show this implies it holds when $\mathcal{T}(U)=k+1$.
When $\mathcal{T}(U)=k+1$ we consider 
\begin{align}
M_{k+1}=\prod_{i=k+1}^{1}\widehat{R(P_{i})} & =\widehat{R(P_{k+1})}M_{k}\label{eq:k+1_case}
\end{align}
where $M_{k}=\prod_{i=k}^{1}\widehat{R(P_{i})}$. Consider the lower-right
$3\times3$ submatrix of $M_{k}$. Using the inductive hypothesis,
the row labeled $P_{k}$ has $\mathrm{sde}$ $k-1$, and the other
two rows have $\mathrm{sde}$ $k$. Since $P_{k+1}\neq P_{k}$, one
of these two rows is the one labeled $P_{k+1}$ and we write $Q=\{X,Y,Z\}\setminus\{P_{k},P_{k+1}\}$
for the other one. Note (looking at equation (\ref{eq:three_matrices}))
that left multiplying $M_{k}$ by $\widehat{R(P_{k+1})}$ leaves the
row $P_{k+1}$ alone but replaces rows $P_{k}$ and $Q$ with $\frac{1}{\sqrt{2}}$
times their sum and difference (in some order). The row $P_{k+1}$
of $M_{k+1}$ therefore has $\mathrm{sde}$ equal to $k$, as required.
Using the fact that row $P_{k}$ of $M_{k}$ has $\mathrm{sde}$ $k$
and row $Q$ of $M_{k}$ has $\mathrm{sde}$ $k-1,$ and using Fact
\ref{Fact}, we see that the two row vectors obtained by taking $\frac{1}{\sqrt{2}}$
times their sum and difference both have $\mathrm{sde}$ $k+1$. Hence
the two rows of $M_{k+1}$ labeled $P_{k}$ and $Q$ have $\mathrm{sde}$
$k+1$. This completes the proof.
\end{proof}

\section{Algorithm for the T-count\label{sec:Algorithm-for-the}}

In this Section we present an algorithm which solves COUNT-T using
$\mathcal{O}\left(N^{m}\text{poly}(m,N)\right)$ time and space.

A naive approach would be to exhaustively search over products of
the form (\ref{eq:clifford_T_expansion}) which alternate Clifford
group elements and T gates, stopping if we find an expression equal
to $e^{i\phi}U$. This does not work very well. Equation (\ref{eq:clifford_T_expansion})
has $m+1$ Clifford group elements and $m$ T-gates (each of which
may act on any out of the $n$ qubits), so the number of products
of this form is 
\[
|\mathcal{C}_{n}|^{m+1}n^{m}.
\]
A better approach is use Proposition \ref{decomposition} and search
over expressions of the form
\begin{equation}
\left(\prod_{i=m}^{1}R(P_{i})\right)C_{0}\label{eq:expression_prodP_C0}
\end{equation}
 until we find one which is equal to $e^{i\phi}U$ for some phase
$\phi$. In this case the size of the search space is $N^{2m}|\mathcal{C}_{n}|$
(recall $N=2^{n}$). However, a small modification improves this algorithm
substantially. Rather than searching for expressions of the form (\ref{eq:expression_prodP_C0})
until we find one that is equal to $e^{i\phi}U$, we can search over
expressions 
\[
U^{\dagger}\prod_{i=m}^{1}R(P_{i})
\]
until we find one that is a global phase times an element of the Clifford
group $\mathcal{C}_{n}$. This reduces the size of the search space
to $N^{2m}$. The goal of the rest of this Section is to describe
a more complicated algorithm which obtains (roughly) a square-root
improvement at the expense of increasing the memory used. 

We work in the channel representation, i.e., we consider the group
$\widehat{\mathcal{J}_{n}}$. To describe our algorithm it will be
convenient to have a notion of equivalence of unitaries \textit{up
to right-multiplication by a Clifford}. To be precise, we consider
the left cosets of $\widehat{\mathcal{C}_{n}}$ in $\widehat{\mathcal{J}_{n}}$.
We now show how to compute a (matrix-valued) function which tells
you whether or not two unitaries $W,V\in\widehat{\mathcal{J}_{n}}$
are from the same coset, i.e., whether or not $W=VC$ for some $C\in\widehat{\mathcal{C}_{n}}.$
\begin{defn}
[\textbf{Coset label}]\label{canonical form}

Let $W\in\widehat{\mathcal{J}_{n}}$ . We define $W^{(c)}$ to be
the matrix obtained from $W$ by the following procedure. First rewrite
$W$ so that each nonzero entry has a common denominator, equal to
$\sqrt{2}^{\mathrm{sde}(W)}$ (recall $\mathrm{sde}$ is defined for
matrices in Definition \ref{sde_defn}). For each column of $W$,
look at the first nonzero entry (from top to bottom) which we write
as $v=\frac{a+b\sqrt{2}}{\sqrt{2}^{\mathrm{sde}(W)}}$. If $a<0$,
or if $a=0$ and $b<0$, multiply every element of the column by $-1$.
Otherwise, if $a>0$ or $a=0$ and $b>0$, do nothing and move on
to the next column. After performing this step on all columns, permute
the columns so they are ordered lexicographically from left to right.
\end{defn}
The following Proposition shows that the function $W\rightarrow W^{(c)}$
labels cosets faithfully, so that two unitaries have the same label
if and only if they are from the same coset.
\begin{prop}
\label{coset_representative_lemma}Let $W,V\in\widehat{\mathcal{J}_{n}}$.
Then $W^{(c)}=V^{(c)}$ if and only if $W=VC$ for some $C\in\widehat{\mathcal{C}_{n}}$.\end{prop}
\begin{proof}
First suppose $W=VC$ with $C\in\widehat{\mathcal{C}_{n}}$. Since
$C\in\widehat{\mathcal{C}_{n}}$ it has exactly one nonzero entry
in every row and every column, equal to either $1$ or $-1$. In other
words $W$ is obtained from $V$ by permuting columns and multiplying
some of them by $-1$. Given this fact, we see from Definition \ref{canonical form}
that $W^{(c)}=V^{(c)}$.

Now suppose $W^{(c)}=V^{(c)}$. From Definition \ref{canonical form}
we see that $W^{(c)}=WD_{1}\pi_{1}$, where $D_{1}$ is diagonal and
has diagonal entries all $\pm1$, and $\pi_{1}$ is a permutation
matrix. Likewise we can write $V^{(c)}=VD_{2}\pi_{2}$. Now setting
these expressions equal we get 
\begin{equation}
W=VC,\quad\text{where}\quad C=D_{2}\pi_{2}\pi_{1}^{-1}D_{1}^{-1}.\label{eq:UVA}
\end{equation}
Note that $C=V^{\dagger}W\in\widehat{\mathcal{J}_{n}}$. Since $C=D_{2}\pi_{2}\pi_{1}^{-1}D_{1}^{-1}$,
we see that it has a single nonzero entry in each row and in each
column, equal to either $+1$ or $-1$. Putting these two facts together
we get $C\in\widehat{\mathcal{C}_{n}}$ (since any unitary in $\widehat{\mathcal{J}_{n}}$
with nonzero entries all $\pm1$ is in $\widehat{\mathcal{C}_{n}}$).
\end{proof}
Our algorithm uses a sorted coset database, defined as follows.
\begin{defn}
[\textbf{Sorted coset database} $\mathcal{D}_{k}^{n}$] \label{sorted_coset_db}For
any $k\in\mathbb{N}$, a sorted coset database $\mathcal{D}_{k}^{n}$
is a list of unitaries $W\in\widehat{\mathcal{J}_{n}}$ with the following
three properties:

\begin{enumerate}[(a)]

\item\textit{Every unitary in the database has $T$-count $k$}\textbf{.}
Every $W\in\mathcal{D}_{k}^{n}$ satisfies $\mathcal{T}(W)=k$.

\medskip{}

\item\textit{For any unitary with $T$-count $k$, there is a unique
unitary in the database with the same coset label.}\textbf{ }For any
$V\in\widehat{\mathcal{J}_{n}}$ with $\mathcal{T}(V)=k$, there exists
a unique $W\in\mathcal{D}_{k}^{n}$ such that $W^{(c)}=V^{(c)}$. 

\medskip{}

\item\textit{The database is sorted according to the coset labels.}
If $W,V\in\mathcal{D}_{k}^{n}$ and $W^{(c)}<V^{(c)}$ (using lexicographic
ordering on the matrices) then $W$ appears before $V$. 

\end{enumerate}

A sorted coset database can be viewed mathematically as a list of
unitaries, but in practice when we implement such a database on a
computer it makes sense to store each unitary along with its coset
label as a pair $\left(W,W^{(c)}\right)$. This ensures that coset
labels do not need to be computed on the fly during the course of
our algorithm.
\end{defn}

\subsection{Algorithm }

We are given as input a unitary $U\in\mathcal{J}_{n}$ and a nonnegative
integer $m$. The algorithm determines if $\mathcal{T}(U)\leq m$,
and if it is, also computes $\mathcal{T}(U)$. 
\begin{enumerate}
\item \textbf{Precompute sorted coset databases $\mathcal{D}_{0}^{n},\mathcal{D}_{1}^{n},\ldots,D_{\left\lceil \frac{m}{2}\right\rceil }^{n}$}.
To generate these databases, we start with $\mathcal{D}_{0}^{n}$
which contains only the $N^{2}\times N^{2}$ identity matrix. We then
construct $\mathcal{D}_{1}^{n}$, then $\mathcal{D}_{2}^{n}$, up
to $D_{\left\lceil \frac{m}{2}\right\rceil }^{n}$ as follows. To
construct $\mathcal{D}_{k}^{n}$ we consider all unitaries of the
form 
\begin{equation}
W=R(P)M,\label{eq:W_hat_loop}
\end{equation}
 where $M\in\mathcal{D}_{k-1}^{n}$ and $P\in\mathcal{P}_{n}\setminus\{\mathbb{I}\}$,
one at a time. We insert $W$ into $\mathcal{D}_{k}^{n}$ (maintaining
the ordering according to the coset labels) if and only if its coset
label is new. That is to say, if and only if no unitary $V$ in one
of the database $\mathcal{D}_{0}^{n},\mathcal{D}_{1}^{n},\ldots,D_{k-1}^{n}$
or in the partially constructed database $\mathcal{D}_{k}^{n}$ satisfies
$W^{(c)}=V^{(c)}$.
\item \textbf{\textcolor{black}{Check if $\mathcal{T}(U)\leq\left\lceil \frac{m}{2}\right\rceil $.}}
Use binary search to check if there exists $W\in\mathcal{D}_{j}^{n}$
for some $j\in\{0,1,...,\left\lceil \frac{m}{2}\right\rceil \}$,
such that $\widehat{U}^{(c)}=W^{(c)}$. If so, output $\mathcal{T}(U)=j$
and stop. If not, proceed to step 3.
\item \textbf{Meet-in-the-middle search.} \textcolor{black}{Start }with
$r=\left\lceil \frac{m}{2}\right\rceil +1$ and do the following.
For each $W\in\mathcal{D}_{r-\left\lceil \frac{m}{2}\right\rceil }^{n}$,
use binary search to find $V\in\mathcal{D}_{\left\lceil \frac{m}{2}\right\rceil }^{n}$
satisfying $\left(W^{\dagger}\widehat{U}\right)^{(c)}=V^{(c)}$, if
it exists. If we find $W$ and $V$ satisfying this condition, we
conclude $\mathcal{T}(U)=r$ and stop. If no such pair of unitaries
is found, and if $r\leq m-1$, we increase $r$ by one and repeat
this step; if $r=m$ we conclude $\mathcal{T}(U)>m$ and stop.
\end{enumerate}
Let us now consider the time and memory resources required by this
algorithm. First consider step $1$. To compute the sorted coset databases
we loop over all unitaries of the form \ref{eq:W_hat_loop}, with
$k\in\{0,\ldots,\left\lceil \frac{m}{2}\right\rceil \}$. There are
$\mathcal{O}\left(N^{2\left\lceil \frac{m}{2}\right\rceil }\right)$
such unitaries since $|\mathcal{P}_{n}|=N^{2}.$ For each unitary
we need to compute the coset label and search to find it in databases
generated so far; since the databases are sorted, the search can be
done using $\mathcal{O}(\text{log}(N^{2\left\lceil \frac{m}{2}\right\rceil }))$
comparisons. The objects we are comparing (the unitaries and their
coset labels) are themselves of size $N^{2}\times N^{2}$, so step
1. takes time $\mathcal{O}(N^{m}\text{poly}(m,N))$. To store the
databases likewise requires space $\mathcal{O}(N^{m}\text{poly}(m,N))$.
Step 2. takes time $\mathcal{O}(\text{poly}(m,N))$ since the binary
search can be performed quickly. In step 3. we loop over all elements
of the databases until we reach the stopping condition; the total
time required in this step is $\mathcal{O}(N^{m}\text{poly}(m,N))$.

\subsection{Correctness of the algorithm}

To prove that the output of our algorithm is correct, we first show
that step 1. of the algorithm correctly generates sorted coset databases
$\mathcal{D}_{0}^{n},\mathcal{D}_{1}^{n},\ldots,D_{\left\lceil \frac{m}{2}\right\rceil }^{n}$.
Given this fact, it is clear that if $0\leq\mathcal{T}(U)\leq\left\lceil \frac{m}{2}\right\rceil $
then step 2. of the algorithm will correctly compute $\mathcal{T}(U)$.
In the following we show that if $\left\lceil \frac{m}{2}\right\rceil <\mathcal{T}(U)\leq m$
then step 3. computes it (and otherwise outputs $\mathcal{T}(U)>m$).

First consider step 1. It is not hard to see that $\mathcal{D}_{0}^{n}$,
which contains only the identity matrix, is a sorted coset database
(since all unitaries with $T$-count $0$ are Cliffords, and every
Clifford has the same coset label as the identity matrix). 

We now use induction; we assume that $\mathcal{D}_{0}^{n},\mathcal{D}_{1}^{n},\ldots,D_{k-1}^{n}$
are sorted coset databases and show this implies that $\mathcal{D}_{k}^{n}$,
as generated by the algorithm, is a sorted coset database. We confirm
properties (a), (b), and (c) from Definition \ref{sorted_coset_db}.
First suppose $W$ is added to the list $\mathcal{D}_{k}^{n}$ by
the algorithm. This implies it has the form 
\begin{equation}
W=\prod_{i=k}^{1}\widehat{R(P_{i})}\label{eq:W_decomp}
\end{equation}
 for some $\{P_{i}\}\in\mathcal{P}_{n}\setminus\{\mathbb{I}\}$, and
that its coset label is not equal to that of some other unitary $V\in\mathcal{D}_{j}^{n}$
with $j<k$. From (\ref{eq:W_decomp}) we see that $\mathcal{T}(W)\leq k$.
Using the inductive hypothesis we see that $\mathcal{T}(W)>k-1$,
since otherwise there would exist $V\in\mathcal{D}_{\mathcal{T}(W)}^{n}$
with $W^{(c)}=V^{(c)}$. Hence $\mathcal{T}(W)=k$ and $\mathcal{D}_{k}^{n}$
satisfies property (a). Now suppose that $Q\in\widehat{\mathcal{J}_{n}}$
has $\mathcal{T}(Q)=k$. Using the fact that $\mathcal{D}_{k-1}^{n}$
is a sorted coset database, there exists $M\in\mathcal{D}_{k-1}^{n}$
and $P\in\mathcal{P}_{n}\setminus\{\mathbb{I}\}$ such that $Q^{(c)}=\left(R(P)M\right)^{(c)}$.
Now looking at the method by which $\mathcal{D}_{k}^{n}$ is generated
we see that there exists a unique $W\in\mathcal{D}_{k}^{n}$ such
that $W^{(c)}=\left(R(P)M\right)^{(c)}$. We have thus confirmed property
(b): for any $Q\in\widehat{\mathcal{J}_{n}}$ with $\mathcal{T}(Q)=k$
there exists a unique $W\in\mathcal{D}_{k}^{n}$ such that $W^{(c)}=Q^{(c)}$.
The ordering of the database $\mathcal{D}_{k}^{n}$ (property (c))
is maintained throughout the algorithm. Hence $\mathcal{D}_{k}^{n}$
is a sorted coset database.

The following Theorem shows that the output computed in step 3. of
the algorithm is correct.
\begin{thm}
Let $U\in\mathcal{J}_{n}$ and $m\in\mathbb{N}$ with $\left\lceil \frac{m}{2}\right\rceil <\mathcal{T}(U)\leq m$,
and let $\mathcal{D}_{0}^{n},\ldots,\mathcal{D}_{\left\lceil \frac{m}{2}\right\rceil }^{n}$
be sorted coset databases. Then $r=\mathcal{T}(U)$ is the smallest
integer in $\{\left\lceil \frac{m}{2}\right\rceil +1,\left\lceil \frac{m}{2}\right\rceil +2,\ldots,m\}$
for which 
\begin{equation}
\left(W^{\dagger}\widehat{U}\right)^{(c)}=V^{(c)}\label{eq:UWVC}
\end{equation}
with $W\in\mathcal{D}_{r-\left\lceil \frac{m}{2}\right\rceil }^{n}$
and $V\in\mathcal{D}_{\left\lceil \frac{m}{2}\right\rceil }^{n}$.\end{thm}
\begin{proof}
Using Proposition \ref{coset_representative_lemma} we see that (\ref{eq:UWVC})
implies 
\[
\widehat{U}=WVC
\]
 for some $C\in\widehat{\mathcal{C}_{n}}$. Hence, whenever \ref{eq:UWVC}
holds we have 
\[
\mathcal{T}\left(U\right)\leq\mathcal{T}(W)+\mathcal{T}(V)=r-\left\lceil \frac{m}{2}\right\rceil +\left\lceil \frac{m}{2}\right\rceil =r
\]
 To complete the proof, we show that \ref{eq:UWVC} holds with $r=\mathcal{T}(U)$.
From Proposition \ref{decomposition} we can write 
\begin{equation}
\widehat{U}=W_{0}V_{0}C_{0}\label{eq:U-0}
\end{equation}
 where $C_{0}\in\widehat{\mathcal{C}_{n}}$ and 
\[
W_{0}=\prod_{i=\mathcal{T}(U)}^{\left\lceil \frac{m}{2}\right\rceil +1}\widehat{R(P_{i})}\qquad V_{0}=\prod_{i=\left\lceil \frac{m}{2}\right\rceil }^{1}\widehat{R(P_{i})}
\]
for some Paulis $P_{i}\in\mathcal{P}_{n}\setminus\{\mathbb{I}\}$.
Note that $\mathcal{T}(W_{0})=\mathcal{T}(U)-\left\lceil \frac{m}{2}\right\rceil $
and $\mathcal{T}(V_{0})=\left\lceil \frac{m}{2}\right\rceil $.

Using the fact that $\mathcal{T}(W_{0})=\mathcal{T}(U)-\left\lceil \frac{m}{2}\right\rceil $
and property (b) from Definition \ref{sorted_coset_db}, there exists
$W\in\mathcal{D}_{\mathcal{T}(U)-\left\lceil \frac{m}{2}\right\rceil }^{n}$
satisfying $W^{(c)}=W_{0}^{(c)}$ , which implies 
\begin{align*}
WC_{1} & =W_{0}
\end{align*}
 for some $C_{1}\in\widehat{\mathcal{C}_{n}}$ by Proposition \ref{coset_representative_lemma}.
Hence 
\[
\widehat{U}=WC_{1}V_{0}C_{0}.
\]
 Now $\mathcal{T}\left(C_{1}V_{0}C_{0}\right)=\mathcal{T}\left(V_{0}\right)=\left\lceil \frac{m}{2}\right\rceil $
so (using the same reasoning as above) there exists $V\in\mathcal{D}_{\left\lceil \frac{m}{2}\right\rceil }^{n}$
satisfying 
\[
C_{1}V_{0}C_{0}=VC_{2}
\]
for some $C_{2}\in\widehat{\mathcal{C}_{n}}$. Hence 
\[
\widehat{U}=WVC
\]
where $C=C_{2}C_{0}$, or equivalently $W^{\dagger}\widehat{U}=VC$.
Applying Proposition \ref{coset_representative_lemma} gives $\left(W^{\dagger}\widehat{U}\right)^{(c)}=V^{(c)}.$
\end{proof}

\subsection{The T-count of Toffoli and Fredkin is $7$}

We implemented our algorithm in C++. For two qubits we were are able
to generate coset databases $\mathcal{D}_{0}^{2},...\mathcal{D}_{6}^{2}$
(which used $3.96$ GB%
\footnote{1 GB=$10^{9}$ bytes%
} of space in total), and for three qubits we generated $\mathcal{D}_{0}^{3},...\mathcal{D}_{3}^{3}$
(size in memory $4.60$ GB); this allows us to run the two-qubit algorithm
with $m=12$ or the three-qubit algorithm with $m=6$ . We ran the
three-qubit algorithm on Toffoli and Fredkin gates with $m=6$, and
we found that the $T$-count of both of these unitaries is $\geq7$.
It was already known that both of these gates can be implemented using
circuits with $7$ T gates \cite{amy2013meet}. Together with our
results, this shows that the $T$-count of Toffoli and Fredkin are
both $7$. In Figures (\ref{fig:Toffoli-circuit}) and (\ref{fig:Fredkin-circuit})
we reproduce the circuits for these gates from reference \cite{amy2013meet},
which are now known to be $T$-optimal. We emphasize that our definition
of $T$-count does not permit ancilla qubits; it may be possible to
do better using ancillae and measurement along with classically controlled
operations (e.g., for the Toffoli gate \cite{Jones2013}). 

\begin{figure}
\includegraphics[scale=0.5]{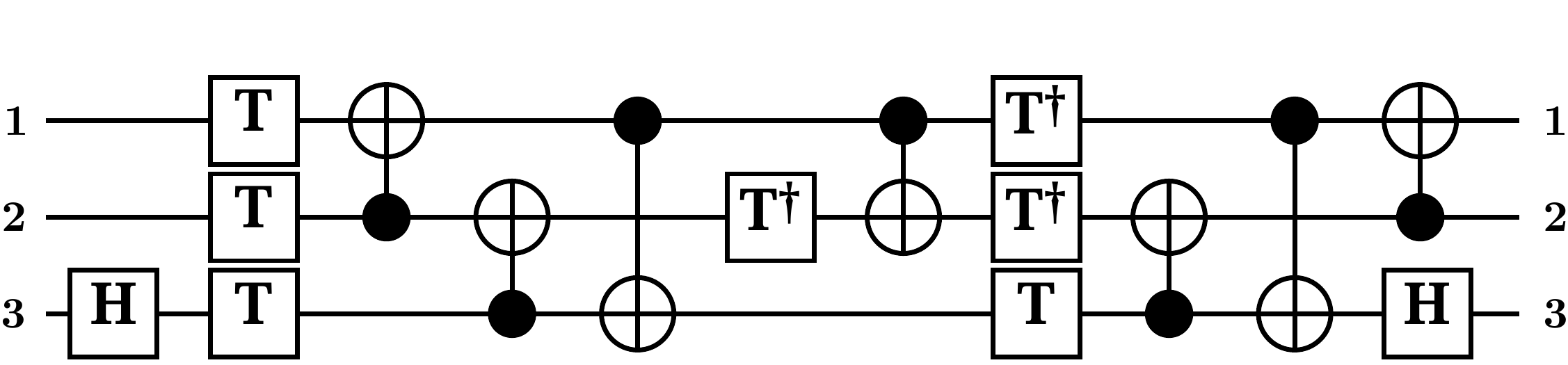}

\caption{Toffoli circuit with 7 $T$ gates \cite{amy2013meet}. Our computer
search shows that this circuit is $T$-optimal: Toffoli cannot be
implemented using 3 qubits with less than $7$ T gates.\label{fig:Toffoli-circuit}}
\end{figure}
 
\begin{figure}
\includegraphics[scale=0.5]{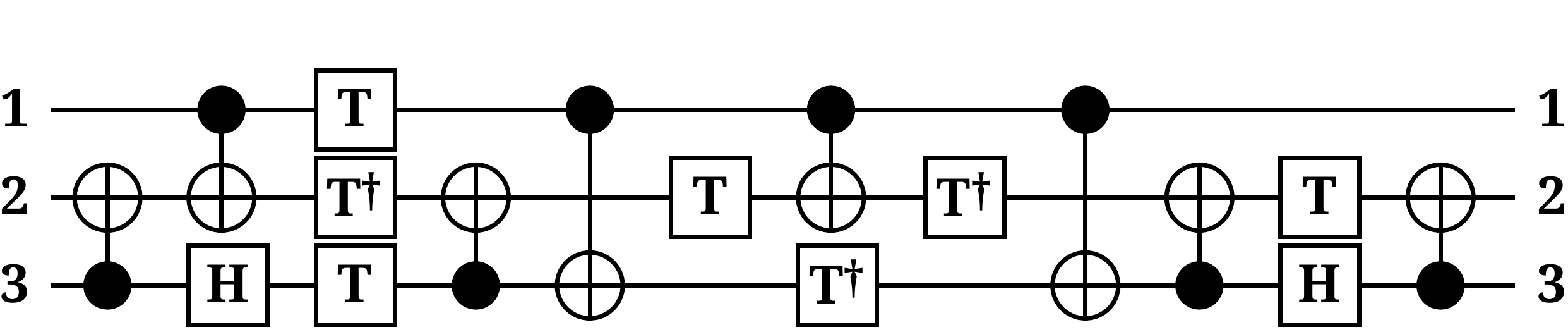}

\caption{Fredkin circuit with $7$ $T$ gates \cite{amy2013meet}. Our computer
search shows that this circuit is $T$-optimal.\label{fig:Fredkin-circuit} }
\end{figure}

\section{Conclusions and open problems\label{sec:Conclusions-and-open} }

Our algorithm for COUNT-T can be viewed as a method of performing
exhaustive search. In contrast, in the single-qubit case we can characterize
the $T$-count as a simple property of the given unitary: the $\mathrm{sde}$
of its channel representation. This characterization does not appear
to generalize to $n>1$ qubits; for example, the $\mathrm{sde}$ of
the channel representation of the Toffoli gate is $2$ but its $T$-count
is 7. 

We conclude by stating the most obvious question, which remains open:
does there exist a polynomial time (as a function of $N$ and $m$)
algorithm for COUNT-T? Alternatively, is this problem computationally
difficult? We do not know the answer to this question even for the
special case of two-qubit unitaries.

\section{Acknowledgments}

\textcolor{black}{We thank Matthew Amy for helpful discussions. }David
Gosset is supported by NSERC.\textcolor{black}{{} }Michele Mosca is
supported by Canada's NSERC, MITACS, CIFAR, and CFI. IQC and Perimeter
Institute are supported in part by the Government of Canada and the
Province of Ontario.

\appendix

\section{Proof of Fact \ref{Clif_Fact}}
\begin{proof}
We first show that it is sufficient to prove the result with $P=Z_{(1)}$.
To see this, consider $P_{A},P_{B}\in\mathcal{P}_{n}\setminus\{\mathbb{I}\}$
and suppose we can efficiently compute circuits for Cliffords $C_{A},C_{B}\in\mathcal{C}_{n}$
which satisfy $C_{A}Z_{(1)}C_{A}^{\dagger}=P_{A}$ and $C_{B}Z_{(1)}C_{B}^{\dagger}=P_{B}$.
Then $C_{B}C_{A}^{\dagger}P_{A}C_{A}C_{B}^{\dagger}=P_{B}$, and we
can efficiently compute the circuit for $C_{B}C_{A}^{\dagger}$.

By conjugating with $\mathcal{O}(n)$ SWAP gates and CNOT gates (both
are Clifford) we can map $P=Z_{(1)}$ into any operator of the form
\begin{equation}
\bigotimes_{i=1}^{n}Z_{(i)}^{y_{i}}\label{eq:prod_zs}
\end{equation}
with $\vec{y}\neq\vec{0}$ an $n$-bit string. This can be seen using
the fact that 
\[
\text{CNOT}\left(\mathbb{I}\otimes Z\right)\text{CNOT}=Z\otimes Z
\]
 and 
\[
\text{SWAP}\left(\mathbb{I}\otimes Z\right)\text{SWAP}=Z\otimes\mathbb{I}.
\]
 Finally, note that the single-qubit Pauli $Z$ matrix can be mapped
to either $X$ or $Y$ by single-qubit Cliffords $H\in\mathcal{C}_{1}$
and $T^{2}H\in\mathcal{C}_{1}$ 
\[
HZH=X\qquad\left(T^{2}H\right)Z\left(T^{2}H\right)^{\dagger}=Y.
\]
 Using these facts it is not hard to see that $P=Z_{(1)}$ can be
transformed into $P^{\prime}\in\mathcal{P}_{n}\setminus\{\mathbb{I}\}$
by first mapping to an operator of the form \ref{eq:prod_zs} by conjugating
with $\mathcal{O}(n)$ CNOTs and SWAPs, and then conjugating by a
tensor product of (at most $n$) single-qubit Cliffords.
\end{proof}
\bibliographystyle{plain}
\bibliography{Tcount_references}

\begin{thebibliography}{10}

\bibitem{aaronson2004improved}
Scott Aaronson and Daniel Gottesman.
\newblock Improved simulation of stabilizer circuits.
\newblock {\em Physical Review A}, 70(5):052328, 2004.

\bibitem{matroid}
M.~{Amy}, D.~{Maslov}, and M.~{Mosca}.
\newblock {Polynomial-time T-depth Optimization of Clifford+T circuits via
  Matroid Partitioning}.
\newblock {\em e-print arXiv: 1303.2042}, March 2013.

\bibitem{amy2013meet}
Matthew Amy, Dmitri Maslov, Michele Mosca, and Martin Roetteler.
\newblock A meet-in-the-middle algorithm for fast synthesis of depth-optimal
  quantum circuits.
\newblock {\em Computer-Aided Design of Integrated Circuits and Systems, IEEE
  Transactions on}, 32:818--830, 2013.

\bibitem{Calderbank}
A.~R. {Calderbank}, E.~M {Rains}, P.~W. {Shor}, and N.~J.~A. {Sloane}.
\newblock {Quantum Error Correction via Codes over GF(4)}.
\newblock {\em eprint arXiv:quant-ph/9608006}, August 1996.

\bibitem{giles2013exact}
Brett Giles and Peter Selinger.
\newblock Exact synthesis of multiqubit clifford+ {T} circuits.
\newblock {\em Physical Review A}, 87(3):032332, 2013.

\bibitem{Gottesman_Heisenberg}
D.~{Gottesman}.
\newblock {The Heisenberg Representation of Quantum Computers}.
\newblock {\em eprint arXiv:quant-ph/9807006}, July 1998.

\bibitem{Jones2013}
C.~{Jones}.
\newblock {Low-overhead constructions for the fault-tolerant Toffoli gate}.
\newblock {\em Physical Review A}, 87(2):022328, February 2013.

\bibitem{kliuchnikov2013synthesis}
Vadym Kliuchnikov.
\newblock Synthesis of unitaries with clifford+{T} circuits.
\newblock {\em eprint arXiv: 1306.3200}, 2013.

\bibitem{kliuchnikov2012fast}
Vadym Kliuchnikov, Dmitri Maslov, and Michele Mosca.
\newblock Fast and efficient exact synthesis of single qubit unitaries
  generated by clifford and {T} gates.
\newblock {\em eprint arXiv: 1206.5236}, 2013.

\bibitem{Veitch2013}
V.~{Veitch}, S.~A. {Hamed Mousavian}, D.~{Gottesman}, and J.~{Emerson}.
\newblock {The Resource Theory of Stabilizer Computation}.
\newblock {\em e-print arXiv: 1307.7171}, July 2013.

\end{thebibliography}

\end{document}